\documentclass[usletter,titlepage,12pt]{amsart}

\usepackage{amsmath,amsthm}\usepackage{natbib}
\usepackage{geometry}
\usepackage{amssymb,authblk,amsfonts}
\usepackage[foot]{amsaddr}
\usepackage{comment,enumerate,graphicx,color,hhline,subcaption}

\usepackage[hidelinks]{hyperref}
\hypersetup{colorlinks = true}
\hypersetup{allcolors=blue}

\def\A{\mathcal{A}} \def\cala{\mathcal{A}}

\def\H{\mathcal{H}}

\def\AH{\mathcal{A}^{\mathcal{H}}}
\def\AC{\mathcal{A}^{\mathcal{C}}}
\def\Re{\mathbf{R}}

\def\ep{\varepsilon}
\def\w{\omega}
\def\one{\mathbf{1}}

\def\dminexp{\ul d}
\def\maxval{\bar \pi}
\def\intoo{\twoheadrightarrow}

\def\argmax{\mbox{argmax}}
\def\supp{\mbox{supp}}

\def\tu{\tilde{u}}

\def\al{\alpha}
\def\la{\lambda}
\def\bla{\Lambda}

\def\os{\emptyset}

\newcommand{\df}[1]{\textit{#1}}
\newcommand{\norm}[1]{\| #1 \|}

\def\ul{\underline}

\definecolor{Blue}{rgb}{1.,0.,0.}

\newdimen\slantmathcorr
\def\oversl#1{
\setbox0=\hbox{$#1$}
\slantmathcorr=\wd0
\hskip 0.2\slantmathcorr \overline{\hbox to 0.8\wd0{%
\vphantom{\hbox{$#1$}}}}
\hskip-\wd0\hbox{$#1$}
}

\def\undersl#1{
\setbox0=\hbox{$#1$}
\slantmathcorr=\wd0
\underline{\hbox to 0.8\wd0{%
\vphantom{\hbox{$#1$}}}}
\hskip-0.8\wd0\hbox{$#1$}
}

\theoremstyle{plain}

\newtheorem{theorem}{Theorem}
\newtheorem{lemma}{Lemma}

\newtheorem{corollary}{Corollary}

\newtheorem*{claim*}{Claim}

\theoremstyle{remark}

\newtheorem{example}{Example}

\newtheoremstyle{named}{}{}{\itshape}{}{\bfseries}{.}{.5em}{#1\thmnote{
    #3}}
\theoremstyle{named}

\newtheoremstyle{named2}{}{}{\itshape}{}{\bfseries}{:}{.5em}{#1\thmnote{
    #3}}
\theoremstyle{named2}

\begin{document}

\author[Echenique]{Federico Echenique}
\author[Miralles]{Antonio Miralles}
\author[Zhang]{Jun Zhang}

\address[Echenique]{Division of the Humanities and Social Sciences, California Institute of Technology}
\address[Miralles]{Department of Economics, Universit\`{a} degli Studi di Messina, Universitat Aut\`{o}noma de Barcelona and Barcelona Graduate School of Economics.}
\address[Zhang]{Institute for Social and Economic Research, Nanjing Audit University, Nanjing, China.}

\title[Fairness and efficiency]{Fairness and efficiency for allocations with participation constraints}

\thanks{We thank Eric Budish, Fuhito Kojima, Andy McLennan, Herv\'e Moulin, and Tayfun S\"onmez for comments.  Echenique thanks the National Science Foundation for its support through the grants SES-1558757 and CNS-1518941. Miralles acknowledges financial support from the Spanish Ministry of Economy and Competitiveness (ECO2017-83534-P,) the Catalan Government (2017 SGR 711,) and the Severo Ochoa Programme (SEV-2015-0563). Zhang thanks the financial support from the National Natural Science Foundation of China (Grant \#71903093).}

\date{May 2020}

\begin{abstract} We propose a notion of fairness for allocation problems in which different agents may have different reservation utilities, stemming from different outside options, or property rights. Fairness is usually understood as the absence of envy, but this can be incompatible with reservation utilities. It is possible that Alice's envy of Bob's assignment cannot be remedied without violating Bob's participation constraint.  Instead, we seek to rule out {\em justified envy}, defined as envy for which a remedy  would not violate any agent's participation constraint.
We show that fairness, meaning the absence of justified envy, can be achieved together with efficiency and individual rationality. We introduce a competitive equilibrium approach with price-dependent incomes obtaining the desired properties.
\end{abstract}



\maketitle

\color{black}
\newpage

\section{Introduction}

We investigate the meaning of fairness in allocation environments with participation constraints and constrained consumption spaces. A special case is the random allocation problem in which agents have unit demand. Without participation constraints, we may regard all agents equally, and the absence of envy is a natural notion of fairness. In our model, different agents may have different reservation utilities, stemming from outside options or property rights. Participation constraints ensure that agents get at least their reservation utilities. Absence of envy may be incompatible with agents' participation constraints. In such environments,  what does it mean to treat agents fairly?

It is well known that allocations satisfying both efficiency and envy-freeness exist \citep{VARIAN197463,HZ1979}. In a model with participation constraints, the challenge is to make efficient and envy-free allocations compatible with agents' individual rationality.  Our contribution is threefold. Our first contribution is to propose a notion of fairness that combines envy and individual rationality. We prove (Theorem \ref{thm:existence}) the existence of fair, efficient, and individually rational allocations. Our second contribution is to show that these fair and efficient outcomes can, under certain conditions, be viewed as {\em market outcomes} (Theorem \ref{thm:NJEeq}), as in Varian and Hylland-Zeckhauser. Our third contribution (Theorem \ref{thm:existenceconstraints}) is to accommodate quantitative constraints, such as those in course allocations (e.g all students must take at least two math courses), or controlled school choice (e.g a school seeks certain diversity objectives).  

We understand fairness as the absence of justified envy, or as ``ruling out envy that can be remedied within agents' individual rationality constraints.'' We do not want to say that an outcome is unfair if its unfairness can be traced to differences in agents' reservation utilities. Concretely, Alice envies Bob at an allocation $x$ if she would rather have Bob's assignment in $x$ than hers. To decide whether this envy is justified, we consider the possibility of swapping the assignments between Alice and Bob, since swapping is an obvious remedy for Alice's envy. We say that Alice's envy is justified if Bob could have received Alice's assignment without violating his participation constraint, and unjustified if Alice's assignment would put Bob below his reservation utility.\footnote{Our fairness notion is similar to the concepts introduced by  \cite{Yilmaz2010} and \cite{athanassoglou2011house} for object allocation problems with ordinal preferences. See Section~\ref{sec:RL} for a discussion. } 

Our notion of envy presumes that the obvious remedy for Alice's envy towards Bob is for them to switch assignments. Clearly, if Alice wants to bring the matter to court, the most natural and plausible remedy she could offer is for the two of them to switch assignments. One might devise more complicated remedies, with a fuller reallocation that would seek to eliminate Alice's envy, but these would necessarily be complicated and require Alice's complaint to rely on multiple agents.  That said, our methods {\em do accommodate more general remedies} (Theorem~\ref{thm:cyclicenvy}).

Importantly, our notion of fairness is compatible with efficiency.
We show that, under some conditions, our solution can be achieved as a market outcome. The idea seeks to generalize Varian's and Hylland and Zeckhauser's competitive equilibrium from equal incomes. The obvious solution would be to endogenize incomes. To this end, we construct price-dependent income functions. We have to be careful since, as shown by  \cite{HZ1979}, when incomes depend on prices, a Walrasian equilibrium might not exist. Our careful construction of income functions ensures individual rationality and fairness. This construction could be regarded as a minimal deviation from equal incomes that sustains individual rationality and no satiated agent overspending. If Alice envies Bob, then Bob's maximum achievable utility is his reservation utility (Lemma \ref{lem:prepenvy}). Besides, if Alice has less money than Bob and she does not envy him, then she has just enough money to reach satiation. We provide an informal description of the income-function construction in Section~\ref{sec:CE}.

We organize the paper as follows. We discuss related literature in Section \ref{sec:RL}, present our model and fairness notion in Section \ref{sec:model}, and present main theorems in Section \ref{sec:main}. We extend our theorems to allocation environments with constraints in Section \ref{sec:constraints}, and show that our theorems can account for more general remedies for envy in Section~\ref{sec:discussion}. In Section \ref{sec:schoolchoice} we apply our result to school choice.

\section{Related literature}\label{sec:RL} Efficiency and fairness can be achieved in models without reservation utilities. Examples are the solutions of \cite{VARIAN197463}, \cite{HZ1979} and \cite{bogomolnaia2001new}. Our problem is complicated, both conceptually and technically, by individual rationality constraints. Conceptually, the meaning of fairness among unequal agents is not obvious, while technically, implementation through market equilibrium may not be possible in economies with price-dependent incomes (see \cite{HZ1979}). Part of our contribution is to support fair, efficient and individually rational (IR) outcomes as competitive pseudo-market equilibria, as in Hylland-Zeckhauser. 

Our notion of no justified envy is analogous to similar notions developed by \cite{Yilmaz2010} and \cite{athanassoglou2011house}. They assume that agents have ordinal preferences instead of cardinal utilities, and say that agent $i$ justifiably envies agent $j$ if $i$ does not regard her allocation as first-order stochastically dominating $j$'s, while any object with positive probability in her allocation is acceptable to $j$.  Yilmaz considers the house allocation with existing tenants model in which some agents have deterministic endowments. He focuses on extending the probabilistic serial rule \citep{bogomolnaia2001new}. Athanassoglou and Sethuraman consider the fractional endowment environment. Their purpose is to extend Yilmaz's mechanism and fairness notion. We work with cardinal utility, focus on market equilibrium instead of probabilistic serial, and use very different techniques. But we share some conceptual similarities with them that extend beyond the similarity in the definition of justified envy. These authors suggest a cake-eating algorithm that starts with all agents eating at the same speed, but  when an agent is at risk of violating her IR constraint, only this agent has the right to eat, until she reaches her reservation utility or drops out of the algorithm. So only when IR binds for some agent is she allowed to eat at higher speed than the others. In our competitive equilibrium method (see Theorem~\ref{thm:NJEeq}), our income functions seek to achieve similar ideas. Fairness pushes us towards equal incomes, but IR forces us to accept some inequality.

\cite{schmeidler1972fair} consider a model where IR constraints arise due to the presence of endowments. Starting from an initial endowment, they study \emph{fair net trades}: trades leading to a Pareto optimal allocation in which no agent envies the trades made by others. Our model differs from theirs for two reasons. First, our model is primarily designed to address constrained consumption spaces, as in the random allocation problem. Fair net trades may not be feasible in such environments, leading to a weak notion of fairness. Specifically, under unit demand constraints there is no reason that one agent's net trade is feasible to any other agent.\footnote{Suppose Bob is endowed with $1/3$ of good 1 and gets all of it; Alice is endowed with $2/3$ of good 1. Then Alice can never envy Bob's net trade, as adding Bob's net trade to her endowment, $2/3+(1-1/3)>1$, violates her unit demand constraint. Schmeidler and Vindt's notion was never meant to be used in (what we term) random allocation problems.} Second, while reservation utilities in our model can arise due to the presence of endowments, they may also stem from other sources. 

In the fractional endowment model, \cite{Zhangfractional} propose algorithms that generalize Top Trading Cycle to organize endowment exchange. Their algorithms find ordinally efficient, fair and individually rational allocations. We differ not only in the use of cardinal utilities but also in the absence of preexisting endowments.

\cite{Kotowski2019exclusion} explore the role of endowments for discrete allocation problems. Different from us, they interpret endowments as the rights to exclude others, and propose a new cooperative game solution concept. Although they allow for public ownership, or collective ownership by subgroups, endowments in their model are deterministic. As a result, their results are unrelated to ours.

Our results are applicable to school choice when we wish to use endowments instead of priorities to control children's rights towards schools. It is in particular applicable to  controlled school choice.  School choice was first introduced as an application of resource allocation models by \cite{abdulkasonmez}. In the standard model of school choice, fairness and efficiency are generally incompatible. A lot of the school choice literature has been devoted to the resulting trade-off. In our solution, the trade-off is resolved.
\cite{hamada2017strategy} is the only paper we are aware of that emphasizes endowments in school choice. They assume that each child owns one seat of some school as endowment. Their goal is to design strategy-proof allocation mechanisms to meet the distributional constraint in the market and  IR constraint of each child. Since they consider deterministic endowments and ordinal preferences, and their fairness notions are based on priorities, their results are unrelated to ours. The constraints we analyze have been discussed extensively in the literature on controlled school choice: see  \cite{ehlers2010}, \citet{kojima2010school}, \cite{hafalir2013effective}, \citet{ehlers2011school}, and \cite{echenique2015control}; and the literature on  distributional constraints (motivated by geographic distributional considerations): \cite{kamada2015efficient}, \cite{kamada2017stability}, and \cite{kamada2019fair}, among others. Our approach of eliminating justified envy when it does not conflict with constraints is common to those papers. For example, Kamada and Kojima consider matchings where no blocking pair that would not violate distributional constraints are present.

In separate work \citep{echenique2019constrained}, we give a direct Walrasian approach to allocation problems with constraints. Key is the idea of setting a price for each constraint. We properly embed constraint prices into each agent's budget constraint.

\section{The model}\label{sec:model}

\subsection{Notation and preliminary definitions.} 

Each agent $i$ has an associated maximum overall demand $c^i\in\Re_{++}$. We define the $c^i$-simplex in $\Re^n$ as $\{x\in\Re^n_+:\sum_{j=1}^nx_j= c^i \}$, denoted by $\Delta^n(c^i)\subseteq\Re^n$. The set $\{x\in\Re^n_+:\sum_{j=1}^nx_j\leq c^i\}$ is denoted by $\Delta_{-}^n(c^i)$. When $n$ is understood, we simply use the notation $\Delta(c^i)$ and $\Delta_{-}(c^i)$. We shall use the shortened notation $C^i=\Delta_{-}(c^i)$ for agent $i$'s \df{consumption space}. When $c^i=1$ for all $i$, we say that agents have \emph{unit demand,} and that we are facing a \emph{random allocation problem}. For $c^i$ large enough for all $i$, the allocation problem is indistinguishable from a standard exchange economy, where the consumption space is $\Re_+^n$.

We adopt the notational conventions used in convex analysis: Denote by $\Re_*=\Re\cup\{-\infty \}$ the extended real numbers. A function $u:\Re^n\to \Re_*$ has \df{domain} $C=\{x\in\Re^n:u(x)>-\infty \}$.\footnote{We use the conventions that are standard in convex analysis. Note that $x+(-\infty)=-\infty$, and $\la \cdot (-\infty)=-\infty$ for any scalar $\la$.}
A function $u:\Re^n \rightarrow \Re$ is 
\begin{itemize}
\item \df{quasi-concave} if $\{z\in \Re^n:u(z)\geq u(x) \}$ is a convex set, for all $x\in \Re^n$.
 \item \df{concave} if, for any $x,z\in \Re^n$, and $\la\in (0,1)$,
  $\la u(z)+(1-\la) u(x)\leq  u(\la z+(1-\la)x)$;
\item \df{linear} if we can identify
  $u(\cdot)$ with a vector $v\in\Re^n$, so that $u(x)=v\cdot x$ on the domain of $u$.  For random allocation problems, linear utility is interpreted as an \df{expected utility function}.
  \item \df{Lipschitz continuous} with constant $\theta>0$ if for any $x,z\in C$, $|u(z)-u(x)|< \theta \norm{z-x}$.
  \item satisfying the \df{Inada property} (at the axes) if, for any $x$ in its domain, $u(x)=u(0)$, unless $x\gg 0$.
  \item has $l$ as its \df{favorite object} if, on its domain, decreasing consumption of any object $k\neq l$ by any amount $\epsilon>0$ in favor of an increased consumption of $\epsilon$ of object $l$ always leads to an increase in $u$.  For example, if $u$ is differentiable and $C=\Re^n_{++}$, then $ l $ is a favorite object if $\frac{\partial u(x)}{\partial
x_{l}}>\frac{\partial u(x)}{\partial x_{k}},$ $\forall k\neq l.$ If $u$ is linear, identified with $v\in\Re^n$, then $ v_l>v_k $, $\forall k\neq l.$
\end{itemize}

Throughout we work with functions that have a closed and convex domain. A function $u$ with domain $C$ is said to be continuous, and monotone if it is continuous (in the relative topology on $C$) at every point $x\in C$, and for any $x,y\in C$, $x<y$ implies that $u(x)<u(y)$.

\subsection{Model} A finite set of agents are to be assigned a finite set of objects. We assume that objects are perfectly divisible. In the random allocation problem, we would be allocating probabilities, and preferences would be defined on the set of probability distributions.

An \df{allocation problem} is a tuple $ \Gamma=\{O, I,Q,(c^i,u^i,\tu^i)_{i\in I}\} $, where:  
\begin{itemize}
\item $ O=\{1,\ldots,L\} $ represents a set of objects, or goods.
\item $I=\{1,\ldots,N\}$ represents a set of agents.
\item $ Q=(q_l)_{l\in O} $ is a capacity vector, and $ q_l\in \Re_{++} $ is the quantity of object $l$. We assume that $\sum_{l\in O} q_l\leq \sum_{i\in I} c^i$ (no overall excess supply.)
\item For each agent $i$, $c^i>0$ determines her maximum overall demand. Her consumption space is $C^i=\Delta_{-}(c^i)$. 
\item For each agent $i$, $u^i:\Re^n\rightarrow \Re_*$ is a continuous and monotone utility function with $C^i$ as its domain. 
\item For each agent $i$, $\tu^i\in \Re$ is her \df{reservation utility}. 
\end{itemize} 

We say that $\Gamma$ admits a \df{common favorite object} if there is an object $l$ that is a favorite object for every agent $i$.

\subsection{Allocations}
An \df{allocation} in $\Gamma$ is a vector $x\in\Re^{LN}_+$, which we write as $x=(x^i)_{i=1}^N$, such that
\[ 
x^i\in C^i \text{ for all }i\in I, \text{ and }\sum_{i\in I} x^i_l= q_l \text{ for all }l\in O.
\] 

Let $ \cala $ be the set of all allocations. Clearly,  $ \cala $ is nonempty, compact and convex.

In the random allocation problem, where $\sum_l q_l= N$, each agent's assignment is a probability distribution over $O$. When $ x^i_l\in \{0,1\} $ for all $i\in I $ and all $ l\in O $, $ x $ is a deterministic allocation. The Birkhoff-von Neumann theorem \citep{birkhoff1946three,von1953certain} implies that every allocation is a convex combination of deterministic allocations.

\subsection{Individual rationality and Pareto optimality.}
An allocation $x$ is \textit{acceptable} to agent $i$ if $u^i(x^i)\ge \tu^i $; $x$ is \df{individually rational} (IR) if it is acceptable to all agents.

We assume that reservation utilities are such that an IR allocation exists. We say that $\Gamma$ admits a \emph{strictly positive IR allocation} if there is an IR allocation $\tilde{x}\in\Re^{LN}_{++}$. All agents obtain strictly positive quantities of all goods in $\tilde{x}$.\footnote{A sufficient condition for this is the existence of an allocation $x$ with $u^i(x^i)>\tu^i$ for all $ i\in I $.}

An allocation $ x $ is \df{weakly Pareto optimal} (wPO) if there is no allocation $y$ such that $ u^i(y^i)> u^i(x^i) $ for all $ i $.
An allocation $ x $ is \df{Pareto optimal} (PO) if there is no allocation $y$
such that $ u^i(y^i)\geq u^i(x^i) $ for all $ i $ with at least one strict inequality.  Given the bounded consumption spaces in our model, wPO is compatible with wasteful situations where one can use existing resources to make some agents strictly better off, but cannot make all agents strictly better off because some agents are satiated.

\subsection{Fairness}
We regard agents as having the right to attain their reservation utilities. While typically reservation utilities arise from endowments, we do not consider endowments as the only source. Guaranteed reservation utilities can arise from a policy that protects disavantaged groups in school choice, for example.

Our notion of fairness rules out envy that cannot be justified by guaranteed reservation utilities. 
If an agent $i$ envies another agent $ j $ at an allocation $ x $ (that is, $ i $ prefers $ x^j $ to $ x^i $), our fairness notion regards the envy as not justified if switching their allocations would violate $j$'s participation constraint.

Formally, we say that an agent $i$ has \df{justified envy} towards another agent $j$ at an allocation $x$ if 
\[
u^i(x^j) > u^i(x^i) \text{ and } u^j(x^i) \ge \tu^j,
\]
and the justified envy is said to be \df{strong} if $ u^j(x^i) > \tu^j $. In words, if $i$ envies $j$ and $j$ could have received $i$'s assignment without violating $j$'s participation constraint, the envy is justified. 

We say that $ x $ has \df{no (strong) justified envy} (N(S)JE) if no agent has (strong) justified envy towards any other agent at $x$.  

Fairness as NJE provides defense against possible complaints. Imagine a social planner that proposes an IR allocation $x$. Suppose an agent $i$ complains that she envies another agent $j$. An obvious remedy for $i$'s complaint is a pairwise switch of their assignments. But if the envy is not justified, the planner's response is that the switch would violate $j$'s right to attain her reservation utility. 

Of course, one may imagine complaints that could be remedied through rearrangements more complicated than a pairwise switch. Such remedies may or may not be realistic, but in any case our methods easily accommodate much more general remedies. Specifically, one can devise cyclic rearrangements, where arbitrarily long sequences of agents collaborate in the satisfaction of an agent's envy, as long as the last agent's participation constraint is not violated. Theorem~\ref{thm:cyclicenvy} in Section~\ref{sec:cyclicenvy} extends our main result to cover this case.

In an IR and NJE allocation $x$, if $ u^i=u^j $ and $\tu^i\ge \tu^j$, then it must be that $u^i(x^i)\ge u^j(x^j) $. That is, if  $ i$ and $j $ have equal preferences and $ i
$'s reservation utility is weakly higher than $ j $'s, then both agree that $ i $'s assignment
in $ x $ is also weakly better than $ j $'s. In particular, if $ u^i=u^j $ and $\tu^i=\tu^j $, then it must be that $ u^i(x^i)=u^j(x^j) $. So NJE and IR imply \textit{equal treatment of equals} (called symmetry by \cite{ZHOU1990123}). 

\subsection{Approximation}  Our main results will prove that there exist allocations that satisfy individual rationality, Pareto optimality, and no justified envy. More precisely, some of our results are based on approximations of these properties: for any $ \ep>0 $, an allocation $ x $ satisfies
\begin{itemize}
	\item  \df{$ \ep
		$-individual rationality} ($ \ep $-IR) if $ u^i(x^i)\ge \tu^i-\ep $ for all $ i\in I $;
	
	\item \df{$\ep$-Pareto optimality} ($\ep$-PO) if there is no
	allocation $y$ such that $ u^i(y^i)> u^i(x^i) +\ep$ for all $ i\in I $.
	
	\item \df{no $ \ep $-justified envy} ($ \ep $-NJE) if there do not exist two distinct agents $ i,j $ such that $ u^i(x^j) > u^i(x^i) $ and $ u^j(x^i) >\tu^j-\ep $.
\end{itemize}

It is clear that $ \ep $-IR is weaker than IR, $ \ep $-PO is weaker than wPO, and $ \ep $-NJE is stronger than NJE.
 
\section{Main results}\label{sec:main}

Let $\Gamma=\{O, I,Q,(c^i,u^i,\tu^i)_{i\in I}\} $ be an allocation problem under the assumptions specified above (that is, utility functions are continuous and monotone, and there exists an IR allocation).

\begin{theorem}\label{thm:existence} Suppose that agents' utility  functions in $\Gamma$ are concave.
	\begin{enumerate}
		\item  For any $ \ep>0 $, there exists an allocation
		that is $ \ep $-individually rational, $\ep$-Pareto optimal and has no $\ep$-justified envy; 
		\item There exists an allocation that is individually
		rational, weak Pareto optimal and has no strong
		justified envy.
		\item\label{it:exists_linear} If utilities are linear, there exists an allocation that is individually rational, Pareto optimal and has no strong justified envy.
\end{enumerate}  
\end{theorem}

Theorem~\ref{thm:existence} gives the existence of allocations with the desired properties of efficiency, fairness, and individual rationality. Our next result gives a competitive market foundation for these allocations.

\begin{theorem}\label{thm:NJEeq} 
	Suppose that agents' utility functions are quasi-concave, and that at least one the following conditions hold: 
	\begin{itemize}
	\item $\sum_{l\in O}q_l<\min_ic^i$ ($c^i$ is sufficiently large for every $i\in I$).
	\item $u^i$ satisfies the Inada property and $\tu^i>u^i(0)$, for every $i\in I$.
	\item There exists a common favorite object, and a strictly positive IR allocation.
	\end{itemize}
Then there exists continuous income functions $m^i:\Delta\rightarrow \Re_+$ and $(x,p) = ((x^i)_{i=1}^N,p)$, such that $p\in\Delta$ is a price vector, and
\begin{enumerate}
  \item $\sum_i x^i = Q$ ($x$ is an allocation; or, ``supply equals demand'').
  \item $x$ is individually rational, Pareto optimal and has no justified envy. 
  \item\label{it:umax}  $x^i\in \argmax\{u^i(z^i): z^i\in C^i \text{ and } p\cdot z^i\leq m^i(p) \}$.
  \end{enumerate}
\end{theorem}

Theorem~\ref{thm:NJEeq} provides conditions under which a fair, efficient, and IR allocation exists \emph{and} can be supported as a form of market equilibrium (a ``pseudo-market''). Our equilibria generalize Hylland and Zeckhauser's, or Varian's, notion of equilibrium with a fixed exogenous income. Here, income is not fixed. It is price dependent, and formulated through \emph{income functions} $m^i$. These are carefully calibrated to ensure both IR and NJE.

Theorem~\ref{thm:NJEeq} allows us to improve on Theorem~\ref{thm:existence}. When utilities are Lipschitz, quasi-concavity is sufficient.

\begin{corollary}\label{cor:lipschitz} If $u^i$ is quasi-concave and Lipschitz continuous for every $i\in I$, there is an allocation that is individually rational, weak Pareto optimal and has no strong justified envy.\end{corollary}

The connection between the two theorems is worth clarifying. We prove the third statement of Theorem~\ref{thm:existence} using Theorem~\ref{thm:NJEeq}. The first two statements of Theorem~\ref{thm:existence} have very different proofs, and can accommodate quantity constraints; See Section~\ref{sec:constraints}.

Finally, the two theorems hold without change if we use a stronger notion of NJE that does not rely on pairwise switch as the remedy for envy; See Section~\ref{sec:cyclicenvy}.

\subsection{Remarks on Theorem~\ref{thm:existence}}\label{sec:rmksThmExistence} Statements (1) and (2) of 
Theorem~\ref{thm:existence} are based on weighted utilitarian maximization. We study the problem of maximizing\begin{equation}
     \sum_{i\in I} \la^i u^i(x^i)\label{eq:weightedut}
\end{equation} over all IR allocations $x$, for each fixed vector $\la=(\la_1,\ldots,\la_N)$ of welfare weights. The trick is to find ``fair'' welfare weights. Ideally, one could proceed iteratively. For each $\la$, solve the weighted utilitarian maximization problem and check if there is any justified envy. If $i$ justifiedly envies $j$, then adjust $\la$ so as to decrease $\la^j$ and increase $\la^i$. Yet the iterative procedure does not quite work. We use a related idea, based on the  \textit{Knaster-Kuratowski-Mazurkiewicz (KKM) lemma}.

The KKM lemma was used by \cite{VARIAN197463} in proving the existence of Pareto-efficient allocations with no envy whatsoever. Varian does not consider participation constraints, and works directly with allocations (more precisely, with the utility possibility frontier). Our approach using welfare weights (inspired by the Negishi approach to equilibrium existence), is quite different. Participation constraints introduce some technical difficulties, which necessitates an approximation argument. The presence of $\ep>0$ in the IR, efficiency and NJE properties are consequences of our approximation argument. 

We briefly explain our use of the KKM lemma. Suppose there is a collection of closed sets, each one identified with a vertex of the simplex. Suppose that each face of the simplex is covered by the union of the sets identified with the vertices of such face. (Notice that the simplex is also a face of itself.) The KKM lemma says that such a collection of sets has non-empty intersection. In the proof of Theorem~\ref{thm:existence}, we identify the simplex with the set of welfare weights $\la$.  Thus each vertex of the simplex is the result of putting all weight on one agent in solving~\eqref{eq:weightedut}. Each set $ \bla^i $ corresponds to the set of weights yielding an $\ep$-Pareto optimal allocation in which 1) agent $i$ does not have $\ep$-justified envy towards any other agent and 2) $ \ep $-individually rationality holds for agent $i$. We show that the collection $ (\bla^i)_{i\in I} $ meets the conditions of the KKM lemma. Any point in the  intersection of $ (\bla^i)_{i\in I} $ meets the properties in the first statement of the theorem.

By taking the limit when $\ep\to 0$, we obtain the second statement of Theorem~\ref{thm:existence}. Observe that we only conclude that the obtained allocation is wPO. This is irrelevant in many allocation problems, in which wPO and PO are identical. There are environments, however, in which there is no such equivalence. Market design problems in which the consumption spaces are bounded (say, because of unit demand) constitute one example.

\subsection{Remarks on Theorem~\ref{thm:NJEeq}}\label{sec:CE}
The use of competitive markets to achieve a fair and efficient allocation is inspired by \cite{VARIAN197463} and \cite{HZ1979}, and more recently by \cite{miralles2014prices}, who establish a Second Welfare Theorem for the kind of allocation problems studied in our paper. 

Varian and Hylland-Zeckhauser use fixed and equal incomes for all agents. The first complication in our paper is that equal incomes will, however, not respect reservation utilities. Incomes must be price-dependent and constructed to satisfy IR and ensure NJE.\footnote{When reservation utility arises from agents' endowments, one may be tempted to use Walrasian incomes. These are, of course, price dependent, and ensure IR. Unfortunately, there are simple examples of allocation problems with endowments where no Walrasian equilibria exist \citep{HZ1979}. See \cite{echenique2019constrained} for a discussion.} Our model suggests a ``minimal departure'' from equal incomes that satisfies IR: We allow an agent to have above-average income only in order to obtain exactly her reservation utility. A second complication is that a competitive equilibrium allocation with potentially satiated agents does not guarantee Pareto-optimality, unless expenses for satiated agents are minimal \citep{HZ1979}. For this reason, we force an agent's income below average whenever the average lies above the minimal income providing her with satiation.

The main new idea in the proof of Theorem~\ref{thm:NJEeq} lies in the construction of price-dependent incomes. The construction is done by, for each price vector, taking the median of three magnitudes: 1) a common income level, 2) the minimum expenditure guaranteeing satiation, and 3) the minimum expenditure ensuring reservation utility. Agents' incomes add up to the overall value of objects. An important property of our construction is that, when $i$'s income is higher than $j$'s, and $j$ is not satiated, then $i$'s income must be equal to the minimum expenditure ensuring her reservation utility. This naturally establishes no justified envy: If agent $j$ envies $i$, then $j$ must have less equilibrium income than $i$, so that $i$ must find $j$'s allocation unacceptable. Finally, it is clear that IR is ensured, since no agent's income lies below the minimum expenditure ensuring her reservation utility.

 Once income functions are in place, existence follows standard ideas: first showing the existence of quasi-equilibrium (as in e.g \cite{mwg1995}, chapter 17, appendix B), and then exploiting the conditions stated in the theorem to bridge the gap between quasi-equilibrium and competitive equilibrium. To this end, the conditions stated in the theorem play a technical role. They serve to ensure that either all prices or incomes are strictly positive for all agents. The first condition consists of virtually unbounded (above) consumption spaces.\footnote{A corollary is that \textit{if the consumption space is $\Re^L_+$ for all agents, then our Walrasian approach obtains an allocation that is Pareto optimal, IR and has no justified envy.}} The second condition considers preferences for strictly positive bundles alongside with reservation utility above the minimal level. The third condition is based on the existence of a common favorite object, and a strictly positive IR allocation. The resulting allocation guarantees all the desired properties in their strongest sense.  

Similarly, Lipschitz continuity in Corollary~\ref{cor:lipschitz} is also a technical condition, ensuring that we can easily create an (arbitrarily) low amount of an artificial favorite good. Existence of a common favorite object is a restrictive assumption. The corollary works by  extending an economy with the addition of an artificial common favorite good. Lipschitz continuity is needed in order to facilitate such inclusion. We take the limit when the supply of the artificial good  tends to zero, obtaining the limit of a sequence of Pareto optimal allocations.

The addition of an artificial favorite good also lies behind the proof of Statement~(\ref{it:exists_linear}) in Theorem~\ref{thm:existence}. With linear (or expected) utility functions, the set of Pareto optimal allocations is closed in the amount of the artificial good. Hence, the limit allocation is also Pareto optimal.

\section{Efficient and fair assignment under constraints}\label{sec:constraints}

Many allocation problems require allocations to satisfy certain quantitative constraints. It is easy to adapt our model and results to such situations. In Section~\ref{sec:constraintthm}, we take  a set $\AC\subseteq \cala$  as the primitive and interpret it as the set of allocations that comply with some given collection of constraints.  $\AC$ is closed and convex, and we implicitly assume that the behind constraints satisfy the condition of \cite{budish2013designing}, which ensures that every feasible allocation can be achieved as randomized deterministic feasible allocations. We discuss several examples of explicit constraints in Section~\ref{sec:structures}.

\subsection{Constrained allocations}\label{sec:constraintthm}
Given a set of feasible allocations $\AC\subseteq \cala$, the definition of individual rationality is same as before. We assume that there exists an IR allocation in $\AC$. The definition of efficiency extends naturally to $\AC$. An allocation $ x\in \AC$ is \df{Pareto optimal} if there is no allocation $y\in  \AC$ such that $ u^i(y^i)\ge u^i(x^i) $ for all $ i\in I $ with strict inequality for some agent; $ x $ is \df{weak Pareto optimal} (wPO) if there is no allocation $y\in  \AC$ such that $ u^i(y^i)> u^i(x^i) $ for all $ i\in I $; and $ x $ is
\df{$\ep$-Pareto optimal} ($\ep$-PO), for any $\ep>0$, if there is no
allocation $y\in  \AC$ such that $ u^i(y^i)> u^i(x^i) +\ep$ for all $ i\in I $. 

Until now, $ i $'s envy towards  $ j $ is negated if switching their assignments violates the participation constraint of $ j $. Now, additional constraints provide another reason for negating $ i $'s envy: switching their assignments may not be feasible because it violates some constraints. To formalize this idea, let $ x_{i\leftrightarrow j} $ denote the allocation obtained by switching the assignments of $i$ and $j$ in an allocation $ x $; that is, $ x^i_{i\leftrightarrow j}=x^j $, $ x^j_{i\leftrightarrow j}=x^i $, and $ x^k_{i\leftrightarrow j}=x^k $ for all $ k\in I\setminus \{i,j\} $.
An agent $i$ has \df{justified envy} towards another agent $j$ at an allocation $x\in \AC$ if 
\[
u^i(x^j) > u^i(x^i), \ u^j(x^i) \ge \tu^j \text{ and } x_{i\leftrightarrow j}\in \AC.
\]

Under the new definition, no justified envy may no longer be compatible with efficiency and individual rationality. To overcome this difficulty, we classify agents into disjoint types. Informally, think of $ i $ and $ j $ as being of equal type if the constraints behind $ \AC $ do not distinguish between them. We identify agents' types by checking whether switching their assignments in any feasible allocation is still feasible. We then prove that fairness among agents of equal type is compatible with efficiency and individual rationality.

Formally, we say two agents $ i,j $ are of \textit{equal type}, denoted by $ i\sim j $, if for all $ x\in \AC $, $ x_{i\leftrightarrow j}\in \AC $.
The binary relation $\sim$ is reflexive and transitive.\footnote{
	Suppose $ i\sim j \sim k $. For all $ x\in \AC $, $ x_{i\leftrightarrow k} =[(x_{i\leftrightarrow j})_{j\leftrightarrow k}]_{i\leftrightarrow j} $. $ i\sim j  $ implies that $ x_{i\leftrightarrow j}\in \AC $, $ j\sim k $ implies that $ (x_{i\leftrightarrow j})_{j\leftrightarrow k}\in \AC$, and $ i\sim j $ implies that $ [(x_{i\leftrightarrow j})_{j\leftrightarrow k}]_{i\leftrightarrow j} \in \AC $. So $ x_{i\leftrightarrow k} \in \AC$.}
Hence it partitions $I$ into disjoint types. 
Then we say $ i $ has \df{equal-type justified envy} towards $j$ at an allocation $x\in  \AC$ if $ i $ has justified envy towards $ j $, and $ i,j$ are of equal type. We say that $ x $ has \df{no equal-type justified envy} if no agent has equal-type justified envy towards any other agent. 
\df{No strong equal-type justified envy} and
\df{no equal-type $ \epsilon $-justified envy} are defined in a similar way by stating that the relevant envy is absent in the allocation. 

With the above definitions, we extend the first two statements of Theorem \ref{thm:existence} to accommodate constraints.
\begin{theorem}\label{thm:existenceconstraints} Suppose that agents' utility functions are concave. 
	\begin{enumerate}
		\item  For any $ \ep>0 $, there exists an allocation
		that is $ \ep $-individually rational, $\ep$-Pareto optimal and has no equal-type $\ep$-justified envy; 
		\item There exists an allocation that is individually
		rational, weak Pareto optimal and has no strong equal-type
		justified envy.
\end{enumerate}
 \end{theorem}

We say that the implicit constraints behind $ \AC $ are \textit{anonymous} if all agents are identified to be of equal type. The model in Section~\ref{sec:model} is one where constraints are anonymous since $ \AC=\A $.

\subsection{Constraint structures.}\label{sec:structures}
It is often most useful to explicitly model the source of types and constraints. For example, types can arise from definitions of socio-economic status, or racial and ethnic classifications. Following the approach of \cite{budish2013designing}, we define a general constraint structure, and then discuss some examples in which an exogenous collection of types gives rise to constrained allocations.

A constraint $ [H,(\ul q_H,\bar q_H)] $ consists of a set  $H\subseteq I\times O$ and a pair of integers $(\ul q_H,\bar q^H)$  with $\ul q_H\leq \bar q_H$. Given a collection $\H$ of constraints, we define 
\[ 
 \AH= \big\{ x\in \A :  \ul q_H\leq  \sum_{(i,l)\in H} x^{i}_{l} \leq \bar q_H \text{ for all } [H,(\ul q_H,\bar q_H)]\in \H \big\}
 \] 
as the set of feasible allocations satisfying $ \H $.

The first example is controlled school choice in which the set of students $ I $ are partitioned into disjoint subsets $T_1,\ldots,T_K$ (which are interpreted as types), and for each school $l$, the desirable number of type $k$ students where $k\in \{1,\ldots,K\}$ is between $ \ul q_{l,k} $ and $ \bar q_{l,k} $. So for each school $l$ and each type $k\in \{1,\ldots,K\}$, we have a constraint  
\[ \big[T_k\times \{l\}, (\ul q_{l,k}, \bar q_{l,k})\big].\]
Theorem~\ref{thm:existenceconstraints} says that there is an efficient and individually rational allocation that achieves fairness within each type.

The second example is the collection of distributional constraints studied by \cite{kamada2015efficient,kamada2017stability}. In that collection, every constraint is of the form
\[
\big[I\times O',(0,\bar q_{O'})\big]
\] where agents are doctors and $ O'\subseteq O $ is the set of hospitals in a geographic region (a city or a prefecture). The collection of constraints is anonymous because each constraint does not distinguish between the identities of doctors. In general, a collection $ \H $ is \textit{anonymous} if for every constraint $ [H,(\ul q_H,\bar q_H)]\in \H $, $ H $ is of the form $ I\times O' $ for some $ O'\subseteq O $. Theorem~\ref{thm:existenceconstraints} implies that for anonymous constraints, there is an efficient and individually rational allocation that achieves fairness among {\em any} two agents.

In the last example, $ \H $ consists of \df{individual constraints}, which impose restrictions on each individual's assignment. In our model of Section~\ref{sec:model}, we have already encountered the $N$ individual constraints:  $[\{i\}\times O,(0,c^i)]$ for each $i\in I$. Formally, $  \H $ consists of $\mathcal{H}_i $ where each $ \mathcal{H}_i$ further consists of constraints of the form 
\[
[\{i\}\times O',(\ul q_{i,O'},\bar q_{i,O'})]
\]
with $O'\subseteq O$. For example, in course allocation if a student $i$ has to take at least one math course but no more than three math courses, we can impose the constraint $[\{i\}\times O',(1,3)]$ where $O'$ is the set of math courses. Theorem~\ref{thm:existenceconstraints} says that we can achieve fairness among agents of equal individual constraints.

\section{Discussion}\label{sec:discussion}

\subsection{Justified envy by exchange}\label{sec:cyclicenvy}
Our notion of NJE relies on pairwise switch as being the remedy for envy. We think of such switch as natural. But if pairwise switch is seen as limited, it is important to note that our result is, in fact, easily generalized to allow for more general remedies.

Let us think of envy that can be addressed by carrying out a chain of exchanges, each agent giving up her assignment in favor of an agent who envies her, and the last agent in the exchange being given the assignment of the first agent. If this reallocation does not violate the last agent's participation constraint, then the envy is justified.

Formally, agent $i$ has \df{justified envy by exchange} towards agent $j$ at allocation $x$ if  there exists a sequence of distinct agents $(i_k)_{k=1}^K$ with \begin{itemize}
	\item $i_1=i$ and $i_2=j$;
	\item $i_k$ envies $i_{k+1}$, $1\leq k\leq K-1$;
	\item and $u^{i_K}(x^{i_1})\ge \tu^{i_K}$.
\end{itemize}
The idea is that $i$ could conjure a remedy for her envy towards $j$ by proposing a coalition of agents and a reallocation of their assignments, such that all are made better off, with the possible exception of one agent whose participation constraint is not violated. We  define \df{strong justified envy by exchange} and \df{$\ep$-justified envy by exchange} similarly as before. We prove that our main theorems hold without change under this extended fairness notion.

\begin{theorem}\label{thm:cyclicenvy}
	Theorem \ref{thm:existence} and Theorem \ref{thm:NJEeq} hold without change if no ($\ep$-/strong)  justified envy is replaced by no ($\ep$-/strong)  justified envy by exchange.
\end{theorem}

\subsection{An example of envy between agents of equal endowment in an allocation of no justified envy} It is natural to think of endowments as a source for reservation utility. In this subsection, we present an example of a discrete allocation problem in which one agent envies another agent in an allocation that is individually rational, Pareto optimal, and satisfies NJE,  even though the two agents have equal endowments.  The punchline is that the two agents have different preferences, so that they play very different roles in the economy. Other agents ``trade'' with them, and as a result one of them ends up being more useful than the other to the remaining agents. The outcome implies the presence of envy.

The example also suggests that our solution may fail to be incentive compatible. We have not specified a selection mechanism, and opted not to discuss incentives and strategy-proofness, but the example conveys some insights. One agent envies another even though they have equal endowments. This fact suggests that one agent may want to pretend to be the agent that he envies. In a large economy in which the number of agents who report each type of preference does not change very much after a misreport, it stands to reason that such a misreport would not be profitable. Of course, the example falls short of proving that if we were to define a fair mechanism it would not be strategy proof.

\begin{example}\label{example:equalenvy} Consider five agents $ \{1,2,3,4,5\} $  and three objects $ \{a,b,c\} $. There are two copies of objects
	$b$ and $c$, but only one copy of object $a$. Agents have linear utility functions. Their von-Neumann-Morgenstern (vNM) utilities and endowments are described in Table \ref{example:equalenvy}. Observe that agents 1 and 2 have identical endowments. 
	\begin{table}[!h]
		\centering
		\begin{subtable}{.3\linewidth}
			\centering
			\begin{tabular}{c|ccc}
				$i$	& $ u^i_{a} $ & $ u^i_{b} $ & $ u^i_{c} $ \\ \hline
				$ 1 $   & 3              & 1                & 2 \\
				$ 2 $   & 3              & 2                & 1 \\
				$ 3 $   & 2 & 3 & 1 \\
				$ 4 $   & 2 & 3 & 1 \\
				$ 5 $   & 2 & 3 & 1 \\
			\end{tabular}
			\subcaption{Utilities}
		\end{subtable}
		\quad
		\begin{subtable}{.3\linewidth}
			\centering
			\begin{tabular}{c|ccc}
				$i$	& $ \w^i_{a} $ & $ \w^i_{b} $ & $ \w^i_{c} $ \\ \hline
				$ 1 $   & 0 & 1 & 0 \\
				$ 2 $   & 0 & 1 & 0 \\
				$ 3 $   & 1/3 & 0 & 2/3 \\
				$ 4 $   & 1/3 & 0 & 2/3 \\
				$ 5 $   & 1/3 & 0 & 2/3 \\
			\end{tabular}
			\subcaption{Endowments}
		\end{subtable}
		\quad
		\begin{subtable}{.3\linewidth}
			\centering
			\begin{tabular}{c|ccc}
				$i$	  & $ x^i_{a} $ & $ x^i_{b} $ & $ x^i_{c} $ \\ \hline
				$ 1 $   & 0                & 0                 & 1 \\
				$ 2 $   & 1/2             & 0                & 1/2 \\
				$ 3 $   & 1/6 & 2/3 & 1/6 \\
				$ 4 $   & 1/6 & 2/3 & 1/6 \\
				$ 5 $   & 1/6 & 2/3 & 1/6 \\
			\end{tabular}
			\subcaption{Allocation $ x $}
		\end{subtable}
		\caption{Example \ref{example:equalenvy}}\label{table:example:equalenvy}
	\end{table}
	
	Consider the allocation $x$ in the same table. Agent 1 envies agent 2 at $ x $ because \[
	u^1\cdot x^1 = 2< 3/2+2/2=u^1\cdot x^2.\] The envy is not justified,
	however, because \[
	u^2\cdot x^1 = 1< 2 = u^2\cdot \w^2.\] In fact, it is easy to see that
	$x$ has no justified envy, and is individually rational and
	Pareto optimal. In any PO allocation $y$, we cannot have $y^1_{b}>0$, as agent $1$ and any agent $j\in\{3,4,5\}$ are willing to trade $b$ for any other object. So $y^1$ must be a convex combination of $(1,0,0)$ and
	$(0,0,1)$. To make agent 1 better off then we would need to give agent
	1 some shares of $ a $, but these can only come at the expense of
	agent 2. To make agent 2 better off, she would need to get more shares
	of $ a $, but these can only come at the expense of agents 3, 4 and
	5. These agents could only exchange shares of $ a $ for shares of $ b $, which agent 2 does not have. All agents rank
	objects $ a $ and $ c $ in the same way.  
\end{example}

\section{Application to school choice}\label{sec:schoolchoice}
School choice is the problem of allocating children to schools when we want to take into account children's (or their parents') preferences \citep{abdulkasonmez}. In the last 15 years, several large US school districts have  implemented school choice programs that follow economists' recommendation and are based on economic theory.\footnote{Boston  \citep{abdulkasonmez,abdulkadirouglu2005boston}, New York \citep{NYCmatch}, and Chicago \citep{pathak2013school} are leading examples.} Practical implementation of school choice programs presents us with a number of lessons and challenges. 

The first lesson is that school choice should be guided by fairness, or lack of justified envy. When given the choice of implementing either a fair or an efficient outcome, school districts have consistently chosen fairness \citep{abdulkadirouglu2005boston,NYCmatch}. One reason could be that district  administrators are concerned with litigation: If Alice envies Bob's school,  then the district can invoke justified envy to argue as a defense that Bob had a higher priority than Alice at the school.\footnote{Observe that this notion of justified envy is pairwise, as is ours.} It is also likely that district administrators, and society as a whole, have an intrinsic preference for fairness, and the preference is strong enough to outweigh concerns over efficiency.

The second lesson is that school districts want to give some children certain rights, like the right to attend a neighborhood school, or the right to go to the same school as an older sibling. In the current practice, such rights are achieved by giving children different priorities. Priorities seem simple, but they are not transparent: Priorities do {\em not} translate immediately into allocation outcomes. Alice may have a good priority in one school, but her chance of getting into the school depends on  all students' choices and priorities, not only on her priority at the school. This is especially true when priorities are coarse, which is common in practice.\footnote{In school choice with coarse priorities, we can prove that it is computationally hard to determine if a given student will be assigned a given school.}
Also, the ordering of students in a priority may not reflect their ``rights'' at the school. Below we construct an example in which if two students switch priorities at one school, the student who climbs up in the priority ends with a worse outcome in the student-optimal stable matching.

\begin{example}\label{example:priorityharm}
 Consider three schools $\{a,b,c\}$ and three children $\{1,2,3\}$. Suppose that school priorities, and children preferences are as follows.

\begin{center}
	\begin{tabular}{ccc}
$a$  & $b$ & $c$ \\ \hline
$2$ & $2$   & $3$ \\
$3$ & $3$   & $1$ \\
$1$ & $1$   & $2$ \\
	\end{tabular} \;\;\;\;\;\;
	\begin{tabular}{ccc}
$1$  & $2$ & $3$ \\ \hline
$b$ & $a$   & $a$ \\
$c$ & $b$   & $c$ \\
$a$ & $c$   & $b$ \\
	\end{tabular} \end{center}

Then the student-optimal stable matching is 
$\mu(1)=b$,  $\mu(2)=a$ and  $\mu(3)=c$.
If $1$ and $2$ switch roles in the priority ranking of school $a$, then the  student-optimal stable matching becomes 
$\mu(1)=c$,  $\mu(2)=b$ and  $\mu(3)=a$. So $ 1 $ becomes worse off. The trick here is that when $1$ and $2$ switch roles, they also change their positions relative to $3$. The message of the example is that justified envy, a ``pairwise'' concept in school choice, cannot rely on the relative position in schools' priorities of the pair in question.
\end{example}

The third lesson is that school districts have demonstrated a strong preference for controlling the racial and socio-economic composition of their schools: so-called {\em controlled school choice}. A common critique of existing school choice programs is that they have led to undesirable school compositions. For example, in Boston, schools have been left with too few neighborhood children,
which has motivated a move away from the system recommended by economists (\citealp{durboston}). In New York City, the new school choice system exhibits high degrees of racial segregation. Segregation in NYC schools is  not new, but the complaint is that the new school choice program may have made it worse, and certainly has not helped. In the words of a recent New York Times article ``\ldots school choice has not delivered on a central promise: to give every student a real chance to attend a good school. Fourteen years into the system, black and Hispanic students are just as isolated in segregated high schools as they are in elementary schools --- a situation that school choice was supposed to ease.''\footnote{``The Broken Promises of Choice in New York City  Schools'',	{\em New York Times}, May 5th, 2017.} The article points to a dissatisfaction with school composition, and access to the best schools.

The situation in NYC has reached a point where there are talks of doing away with school priorities, and instead instituting a lottery. In fact, Professor Eric Nadelstern at Columbia University, who served as deputy school chancellor when the new school choice system was implemented, has recently proposed that children be allowed to apply to any school, and have a lottery deciding the allocations.\footnote{``Confronting Segregation in New York City  Schools'', {\em New York Times}, May 15th, 2017.}

Given the absence of a direct connection between priorities and outcomes, and the situation in NYC, we propose the use of endowments to control children's rights. This makes Nadelstern's proposal compatible with school choice. We imagine that there is a lottery that gives an {\em initial} probabilistic allocation of children to schools. The lottery could be as simple as giving each child the same chance of attending any school. It could also reflect different objectives in controlled school choice, such as giving each child a higher chance of attending his or her neighborhood school, or giving each minority child a chance (literally, a positive probability) of attending the highest-ranked schools. The initial allocation, or endowment, provides transparent and immediate reservation utility. A child who is endowed with a seat at her neighborhood school can simply choose to attend that school. His or her right to attend that school does not depend on other children in any way.

The initial allocation is typically not the final allocation, because we want preferences to play a role. We assume that children use expected utilities to compare lotteries and ask them to report vNM utilities of schools. For convenience, we may use normalization by  requiring each child to assign utility 1 to his favorite school and assign utility 0 to the worst school. To capture the desired bounds on the composition of a school, we can use quantity constraints. Subject to constraints, our solution achieves all the desirable properties. The final allocation will be fair, efficient, and individually rational.

Beyond school choice, our model and results can apply to other market design problems. An example is time bank where agents exchange labor without using transfer (see \cite{andersson2018organizing}). In that problem every agent demands others' services and also provides services to others. The services an agent can supply are her endowments, and define her reservation utility.

\section{Proof of Theorems~\ref{thm:existence} and \ref{thm:existenceconstraints}}\label{sec:proofexistence}

We prove Theorem~\ref{thm:existenceconstraints} in this section. The first two statements of Theorem~\ref{thm:existence} are corollaries. We prove the third statement of Theorem~\ref{thm:existence} in Section \ref{sec:prooflipschitz} after proving Theorem~\ref{thm:NJEeq} in the next section.


For any given $ \ep>0 $, define
\[\A^* =\{ x\in  \AC: x \text{ is $ \ep $-individually
	rational}\}.\]
It is easy to see that $\A^*$ is nonempty and compact. 

Let the $N$-dimensional simplex $\Delta$ be the domain of welfare weights. 

For any $ \la\in \Delta$, define 
\[ 
\phi(\la)= \argmax\{
\sum_{i\in I} \la^i u^i(x^i) - \delta \sum_{i\in I} \norm{x^i -
	\one}: (x^i)_{i\in I}\in \A^* \}, \] 
where $\one$ is a vector of ones and $ \delta>0$ is small enough such that
\[
\delta\max_{x\in \A^*}\sum_{i\in I} \norm{x^i - \one}<\ep.
\]

Since all $ u^i $ are continuous and concave and $ \sum_{i\in I}
\norm{x^i - \one} $ is continuous and strictly convex, the
objective function $ \sum_{i\in I} \la^i u^i(x^i) - \delta
\sum_{i\in I} \norm{x^i - \one} $ is continuous and strictly
concave. Moreover, $\A^*$ is compact.  Thus, $\phi:\Delta\rightarrow \A^*$ is a function
(meaning it is singleton-valued), and, by the Maximum Theorem, continuous. Moreover, the choice of $\delta$ implies that $\phi$ is $\ep$-Pareto optimal.

For any agent $ i $, define
\[
\bla^{i} = \{\la\in\Delta : \nexists j\in I \text{ s.t  }i \text{ has equal-type }\ep \text{-justified envy towards }j \text{ at } \phi(\la) \}.
\]

The proof relies on an application of the so-called KKM Lemma (the
lemma is due to Knaster, Kuratowski and Mazurkiewicz; see Theorem 5.1 in 
\cite{border1989fixed}). In the following two lemmas we prove that $
\{\bla^{i} \}_{i=1}^N $ is a \df{KKM covering} of the simplex $ \Delta
$. This means that every $\bla^i$ is closed and that for any $\la\in \Delta$ there is at least one $\bla^i$ such that $\la^i>0$ and $ \la\in \bla^i $.

\begin{lemma}\label{lem:Ciclosed}
	For every $ i\in I $, $\bla^{i}$ is closed.
\end{lemma}
\begin{proof}
	Let $\la_n$ be a sequence in $\bla^{i}$ such that $\la_n\rightarrow
	\la\in\Delta$. Let $x_n= \phi(\la_n)$. By continuity of $ \phi $,
	$x_n\rightarrow  x=\phi(\la)\in\A^*$. Now we prove that $\la\in 
	\bla^{i}$, that is, $ i $ does not have equal-type  $\ep$-justified envy
	towards any other agent. Suppose that there is an agent $j$ of equal type with $i$ such that $u^i(x^j)>  u^i(x^i)$ and $u^j(x^i)> \tu^j-\ep$. Since $i$ and $j$ are of equal type, $ x_{i\leftrightarrow j}\in\AC$, and  $(x_{n})_{i\leftrightarrow j}\in \AC$ for every $ n $. By continuity of $u^i$ and $ u^j $, for $n$ large enough we have
	$u^i(x^j_n)>  u^i(x^i_n)$ and $u^j(x^i_n)>
        \tu^j-\ep$. These mean that $ i $ has equal-type $\ep$-justified envy towards $ j $ at 
	$x_n$, which is a contradiction. Therefore, $ \la\in \bla^{i}  $ and $\bla^i$ is closed.
\end{proof}

\begin{lemma}\label{lem:KKMconditions}
	For every $ \la\in \Delta $, $ \la \in \cup_{i\in \supp(\la)} \bla^i $.
\end{lemma}
\begin{proof}
	Suppose, towards a contradiction, that for some $ \la \in \Delta$,
	$\la \notin \cup_{i\in \supp(\la)} \bla^i  $. Let $x=\phi(\la)$. Then
	for every $i\in \supp (\la)$ there exists some $j$ of equal type with $ i $ such that	$u^i(x^j)>u^i(x^i)$ and $u^j(x^i)>\tu^j-\ep$. 
	
	Suppose first that there exist some $i$ and $j$ in the
	aforementioned situation such that
	$j\notin \supp (\la)$. Then consider the allocation $ y
        =x_{i\leftrightarrow j}\in \AC$. $ y $ is $\ep$-individually
	rational as $x$ was $\ep$-individually rational and
        $u^j(x^i) > u^j(\w^j)-\ep$. Note that $\la^j=0$ and $u^i(x^j)>u^i(x^i)$ imply that $\sum_{h\in I} \la^h
        u^h(x^h) < \sum_{h\in I}  \la^h
	u^h(y^h)$. We also have that  $\sum_{h\in I} \norm{x^h - \one} = \sum_{h\in
		I} 
	\norm{y^h -\one}$, hence \[
	\sum_{h\in I} \la^h u^h(x^h) -\delta \sum_{h\in I} \norm{x^h- \one}
	< \sum_{h\in I} \la^h u^h(y^h) -\delta  \sum_{h\in I} \norm{y^h
		-\one}.\] 
	But it contradicts the definition of $ x=\phi(\la)$.
	
	The above argument means that every $i\in \supp (\la)$ has equal-type $\ep$-justified 
	envy towards some $ j\in \supp (\la) $. Then, since the set of agents
	in $\supp(\la)$ is finite, there must exist a subset of distinct agents $\{i_1,\ldots i_K\}\subseteq \supp(\la)$ such that $i_1$ has equal-type $ \ep $-justified envy
	towards $i_2$, $i_2$ has equal-type $ \ep $-justified envy towards
	$i_3$, and so on until $i_K$ has equal-type $ \ep $-justified envy
	towards $i_1$. Then we can construct a new allocation $ y $ by letting
	agents in the cycle exchange their allocations. Since the
        agents in the cycle are of equal type, $ y $ must be
        feasible, that is, $ y\in \AC $.\footnote{We can consider a sequence of allocations $\{x(k)\}_{k=0}^{K-1}$ with $x(0)=x$ and $x(k) =
          x_{i_{k}\leftrightarrow i_{k+1}}(k-1)$ for each $ 1\le k\le K-1 $. Since all agents in the cycle are of equal type, each  $x(k)\in  \AC$. We let $y=x(K-1)$.} 
      As before, we have that  $\sum_{h\in I} \norm{x^h - \one} = \sum_{h\in I} \norm{y^h-\one}$ because $y$ is obtained from $x$ by permuting the assignments of agents in the cycle. Then we have
	\[
	\sum_{h\in I} \la^h u^h(x^h) -\delta \sum_{h\in I} \norm{x^h - \one}
	< \sum_{h\in I} \la^h u^h(y^h) -\delta  \sum_{h\in I} \norm{y^h -\one}.\] 
	
	As before, it is a contradiction. 
\end{proof}

Now we are ready to prove Theorem~\ref{thm:existenceconstraints}.

\begin{proof}[\color{blue}Proof of Theorem~\ref{thm:existenceconstraints}\color{black}]
	The proof is an application of the KKM lemma: see Theorem 5.1 in
	\cite{border1989fixed}.  
	
	By Lemmas~\ref{lem:Ciclosed} and~\ref{lem:KKMconditions}, $ \{\bla^i
	\}_{i=1}^n $ is a KKM covering of $ \Delta $. So there exists
	$\la^*_\ep\in\cap_{i=1}^n \bla^i$. Let $x^*_\ep=\phi(\la^*_\ep)$. Then
	$x^*_\ep$ is $ \ep $-individually rational, $ \ep $-Pareto optimal and
	has no equal-type $ \ep $-justified envy. 
	
	Now let $ \{\ep_n\} $ be a sequence such that $ \ep_n>0
	$ for all $ n $ and $ \ep_n\rightarrow 0 $. Let $ x^*_n $ be the
	allocation found above for each $ \ep_n $. Since the sequence $
	\{x^*_n\} $ is bounded, it has a subsequence $ \{x^*_{n_k}\} $ that
	converges to some $ x^* $. Since the set of feasible allocations is closed, $
	x^* $ is a feasible allocation. We prove that $ x^* $ is individually
	rational, weak Pareto optimal and has no strong equal-type justified envy. 
	
	Since $ u^i(x^{*i}_{n_k})\ge \tu^i-\ep_{n_k}$ for all $ n_k $ and
	all $ i $, in the limit $ u^i(x^{*i})\ge \tu^i  $ for all $ i
	$. So $ x^* $ is individually rational. Suppose $ x^* $ is not weak
	Pareto optimal, then there exists a feasible allocation $ y $ such that $
	u^i(y^i)>u^i(x^{*i}) $ for all $ i $. For big enough $ n_k $, $
	u^i(y^i)>u^i(x^{*i}_{n_k})+\ep_{n_k} $ for all $ i $, which
	contradicts the $ \ep_{n_k} $-Pareto optimality of $ x^*_{n_k} $.
	Suppose some agent $ i $ has strong equal-type justified-envy towards another agent $j$ in $x^* $; that is, $ u^i(x^{*j})> u^i(x^{*i})  $ and $u^j(x^{*i})>\tu^j $. Then for big enough $
	n_k $, $ u^i(x^{*j}_{n_k})> u^i(x^{*i}_{n_k})  $ and $
        u^j(x^{*i}_{n_k})>\tu^j-\ep_{n_k} $. But given that  $i$ and
        $j$ are of equal type, this contradicts the property of no
        equal-type $\ep_{n_k}$-justified envy of $ x^*_{n_k} $. 
\end{proof}

\section{Proof of Theorem~\ref{thm:NJEeq}}\label{sec:proofNJEeq}

We let the $L$-dimensional simplex $\Delta^L $ be the domain of prices.

\subsection{Incomes}\label{sec:incomes}
The key to the theorem is to carefully construct price-dependent income functions. For each consumer $i$, define $i$'s \textit{expenditure function} as 
\begin{equation*}
e^i(v,p) = \inf \{p\cdot x : u^i(x)\geq v \}, 
\end{equation*}
for $p\in\Delta^L$ and $v\in\mathbf{R}$.

Let $v^i = \sup u^i(C^i)$ be the utility of agent $i$ when she is satiated.

For any scalar $m\geq 0$ and $p\in \Delta^L$, let 
\begin{equation*}
\mu^i (m,p) = \mbox{median}(\{e^i(\tu^i,p),m, e^i(v^i ,p) \}).
\end{equation*}

Consider the function 
\begin{equation*}
\varphi(m,p) = \sum_i \mu^i(m,p) - p\cdot Q. 
\end{equation*}

Observe that
\begin{itemize}
\item $e^i(\tu^i,p) \leq e^i(v^i ,p) $.

\item $\mu^i$ is continuous and $m\mapsto \mu^i(m,p)$ weakly monotone
increasing.

\item $\varphi$ is continuous and $m\mapsto \varphi(m,p)$ weakly monotone
increasing.

\item $\varphi(m,p)\leq 0$ for $m\geq 0$ small enough as $%
\sum_ie^i(\tu^i,p)\leq p\cdot Q$ (since an IR allocation $x$ exists, $e^i(\tu^i,p)\leq p\cdot x^i$ for all $i$, and $\sum_ip\cdot x^i=p\cdot Q$.)
\end{itemize}

We shall define $m^i(p)$. First, in the case that $\sum_i e^i(v^i, p) <
p\cdot Q$, we let $m^i(p) = e^i(v^i, p)+\frac{1}{N}[p\cdot Q-\sum_i e^i(v^i, p)]$.  Second, in the case that $\sum_i
e^i(v^i, p) \geq p\cdot Q$, we have that $\varphi(m,p)\leq 0$
for $m\geq 0$ small enough, and $\varphi(m,p)\geq 0$ for $m\geq 0$ large
enough. Therefore there exists $m^*\geq 0$ with $\varphi(m^*,p)=0$.

Now let $m^i(p) = \mu^i(m^*,p)$. To show that this is well defined, we need
to prove that $m^i(p)$ is independent of the choice of $m^*$. To that end,
suppose that there are $m_1,m_2\in\mathbf{R}_+$ with $m_1\neq m_2$ and $%
0=\varphi(m_1,p) = \varphi(m_2,p)$. Suppose without loss of generality that $m_1<m_2$. Now, since
each $\mu^i$ is weakly monotone increasing as a function of $m$ we must have 
$\mu^i(m_1,p) = \mu^i(m_2,p)$ for all $i$. Then the definition of $m^i(p)$
is the same regardless of whether we choose $m_1$ or $m_2$.

Note that, in all cases, $p\cdot Q = \sum_i m^i(p)$.

\begin{lemma}\label{lem:micont}
$m^i$ is continuous.
\end{lemma}

\begin{proof}
Let $p^n\rightarrow p\in \Delta^L$. Note that if $\sum_i e^i(v^i,p) -
 p\cdot Q<0$, then for $n$ large enough we will have $\sum_i
e^i(v^i,p^n) -  p^n\cdot Q<0$. Then $m^i(p^n) =
e^i(v^i, p^n)+\frac{1}{N}[p^n\cdot Q-\sum_i e^i(v^i, p^n)]\rightarrow e^i(v^i, p)+\frac{1}{N}[p\cdot Q-\sum_i e^i(v^i, p)]=m^i(p)$, by continuity of the expenditure
function.

So suppose that $\sum_i e^i(v^i,p) - p\cdot Q\geq 0$, and let $m$ be such that $\varphi(m,p)=0$. We shall discuss two cases.

Case1: Consider the case that $\sum_i e^i(v^i,p^{n_k}) - p^{n_k}\cdot
Q < 0$ for some subsequence $p^{n_k}$. Then $\sum_i e^i(v^i,p) -
 p\cdot Q = 0$. This means that if $\varphi(m,p)=0$ then $m\geq
e^i(v^i,p)$ for all $i$. Hence $m^i(p)=e^i(v^i,p)$ for all $i$. But since $m^i(p^{n_k})
= e^i(v^i,p^{n_k})+\frac{1}{N}[p^{n_k}\cdot Q-\sum_i e^i(v^i, p^{n_k})]$, we get that $m^i(p^{n_k})\rightarrow m^i(p)$.

Case 2: Now turn to a subsequence $p^{n_k}$ with $\sum_i e^i(v^i,p^{n_k}) -
 p^{n_k}\cdot Q \geq 0$. Then there is $m^{n_k}$ with $\varphi(m^{n_k},p^{n_k})=0$. We can take this sequence to be bounded:
consider any further convergent subsequence $m^{n_k^{\prime }}$ and say that 
$m^{n_k^{\prime }}\rightarrow m^{\prime }$. Then $0=\varphi(m^{n_k^{\prime
}},p^{n_k^{\prime }})\rightarrow \varphi(m^{\prime },p)$. Thus $m^i(p^{n_k^{\prime }}) = \mu^i(m^{n_k^{\prime }},p^{n_k^{\prime
}})\rightarrow \mu^i(m^{\prime },p)$, as $\mu^i$ is continuous.
Since the sequence $\{m^{n_k}\}$ is bounded, this implies that $%
m^i(p^{n_k})\rightarrow m^i(p)$.

Cases 1 and 2 exhaust all possible subsequences of $p^n$.
\end{proof}

The role of the following lemma will be clear towards the end of the
proof. 

\begin{lemma}\label{lem:prepenvy}
  If $m^i(p)<\min\{m^j(p),e^i(v^i,p)\}$ then $m^j(p) =
  e^j(\tu^j,p)$. 
\end{lemma}
\begin{proof} Since $m^i(p)<e^i(v^i,p)$, we must be in the case $\sum_i
e^i(v^i, p) \geq  p\cdot Q$ of the definition of income
functions. So let  $m^*\geq 0$ with $\varphi(m^*,p)=0$.

Since $m^i(p) = \mu^i(m^*,p)< e^i(v^i,p)$, we must have $m^*\leq
m^i(p)$. By hypothesis, $m^*<m^j(p)$. Then $m^j(p) = \mu^j(m^*,p)$
implies that $m^j(p) = e^j(\tu^j,p)$.
\end{proof}

\subsection{Existence of quasi-equilibrium}

We first establish the existence of a quasiequilibrium with $p^*\neq
0$. The argument is similar to \cite{galemascolell1975}. See also
\cite{mwg1995} (Chapter 17, Appendix B). 

For any $p\in \Delta^L$, let $\dminexp^{i}(p)$ be the set of
vectors $x^{i\prime }\in C^i$ that satisfy the following properties:
\begin{eqnarray*}
p\cdot x^{i\prime } &\leq &m^{i}(p) \\ 
u^{i}(x^{i\prime }) &\geq &u^{i}(\hat x^{i})\text{ for all }%
\hat x^{i}\in C^i\text{ with }p\cdot \hat  x^{i}<m^{i}(p).
\end{eqnarray*} We consider the correspondence $p\mapsto \dminexp^i(p)$
with domain in $\Delta^L$. 

Observe that $\varnothing \neq \underset{x^{i\prime }\in C^i}{\arg \max }\left\{
u^{i}(x^{i\prime }):p\cdot x^{i\prime }\leq m^{i}(p)\right\} \subseteq
\dminexp^{i}(x,p)$). So $\dminexp^i$ takes non-empty values.

Observe also that $\dminexp^i$ is convex valued. To see this, let
$z^i,y^i\in \dminexp^i(p)$ and define $x^i(\al) = \al z^i+(1-\al)
y^i$, for $\al\in [0,1]$. It is obvious that $x^i(\al)\in C^i$
and that $p\cdot x^i(\al)\leq m^i(p)$.  For any $\hat x^i\in C^i$
with $p\cdot \hat  x^{i}<m^{i}(p)$, $\min\{ u^i(z^i),u^i(y^i)\}\geq
u^i(\hat x^i)$ and quasi-concavity of $u^i$ imply that
$u^i(x^i(\al))\geq u^i(\hat x^i)$. Thus $x^i(\al)\in\dminexp^i(p)$. 

A third observation is that $\dminexp^i(p)$ is upper-hemicontinuous. To
this end, consider a sequence $p_n$ in $\Delta^L$ with $p_n\rightarrow 
p\in \Delta^L$. Consider $z^i_n\in \dminexp^i(p_n)$ such that $z^i_n\rightarrow
z^i$. Clearly, $z^i\in C^i$ and $p\cdot z^i\leq m^i(p)$ as $m^i$
is continuous (Lemma~\ref{lem:micont}). Moreover, for any $\hat x^i\in
 C^i$ with $p\cdot \hat x_i<m^i(p)$, we have that
$p_n\cdot \hat x_i<m^i(p_n)$ for $n$ large enough (again by  Lemma~\ref{lem:micont}).
Thus $u^i(z^i_n)\geq u^i(\hat x^i)$ for $n$ large enough, which by
continuity of $u^i$ implies that $u(z^i)\geq u^i(\hat x^i)$. Hence
$z^i\in\dminexp^i(p)$.

For any $x\in \times_i C^i$ and $p\in\Delta^L$, let \[\maxval(x,p)= \argmax \{p\cdot \left( \sum_i x^i - Q
\right) : p\in \Delta^L\},
\] and  consider the correspondence
\[
\xi: \times_i C^i\times \Delta^L \intoo  \times_i C^i\times \Delta^L \]
defined by 
$\xi (x^1,\ldots,x^N,p) = (\times_i \dminexp^i(p))\times \maxval(x,p)$. 

By the previous observations, and the maximum theorem, $\xi$ is in
the hypotheses of Kakutani's fixed point theorem.  Let $(x^*,p^*)$ be
a fixed point of $\xi$. 

We argue that $(x^{\ast },p^{\ast })$ is a Walrasian quasiequilibrium.
We have that $p^{\ast }\cdot x^{i\ast }\leq m^{i}(p^{\ast })$ for every $i$,
by construction of $\xi$. By definition of $m^{i}$, we have
$\sum_{i}m^{i}(p^{\ast })= p^{\ast }\cdot Q$. Hence, $p^{\ast }\cdot \left( \sum_{i}x^{i\ast }-Q\right) \leq 0.$
This implies $\sum_{i}x^{i\ast }-Q\leq 0$ since otherwise,
by definition of $\maxval$, we would have \[p^{\ast }\cdot \left(
\sum_{i}x^{i\ast }-Q\right) =\underset{p^{\prime }\in
\Delta^L } {\max } \big\{p^{\prime }\cdot \left( \sum_{i}x^{i\ast }-Q\right)  \big\} >0.\]

We show that $\sum_{i}x^{i\ast }-Q=0$. We first consider the case $\sum_ie^i(v^i,p^*)<p^*\cdot Q$. By definition of $m^i$, all agents $i$ have $m^i(p^*)>e^i(v^i,p^*)$ so they must be satiated following the definition of $\dminexp^{i}$. By monotonicity of preferences, we observe $\sum_lx^{i*}_l=c^i$, hence \[ \sum_i\sum_lx^{i*}_l=\sum_ic^i\geq \sum_lq_l\] where the inequality comes from the no overall excess supply assumption. Given that we knew $\sum_{i}x^{i\ast }-Q\leq 0$, we conclude $\sum_{i}x^{i\ast }-Q= 0$. 

We then consider the case $\sum_ie^i(v^i,p^*)\geq p^*\cdot Q$. We claim that $p^{\ast }\cdot x^{i\ast }=m^{i}(p^{\ast })$ for
every $i$, since by definition of $m^{i}$ we have $m^{i}(p^{\ast
})\leq  e^{i}(v^{i},p^{\ast })$. Indeed, suppose that $p^{\ast }\cdot x^{i\ast
}<m^{i}(p^{\ast })\leq e^{i}(v^{i},p^{\ast })$. Since $x^{i\ast }$ does not
satiate the agent, for an arbitrarily small ball $B$ around $x^{i\ast }$\
there is $x^{i\prime }\in B$ with $u^{i}(x^{i\prime })>u^{i}(x^{i\ast })$
and $p^{\ast }\cdot x^{i\prime }<m^{i}(p^{\ast })$, contradicting $x^{i\ast
}\in \dminexp^{i}(x^{\ast },p^{\ast })$. Observe that, as a consequence of
the above,
\begin{equation}
  \label{eq:sumweqsummi}
 p^*\cdot Q = \sum_i m^i(p^*)= \sum_i p^*\cdot x^{i *} .
  \end{equation}

Consequently,  $p^{\ast }\cdot \left(
\sum_{i}x^{i\ast }-Q\right) =0.$ Since $\sum_{i}x^{i\ast
}-Q\leq 0,$ we obtain $p_{l}^{\ast }=0$ for any $l$
with $\sum_i x^{\ast i}_l - q_l<0$ (underdemanded
objects). Then, since preferences are monotonic, it is wlog to assume
that $\sum_{i}x^{i\ast }-Q=0$ by consuming the remaining units
of underdemanded objects for free.

This proves that $(x^{\ast },p^{\ast })$ is a Walrasian quasiequilibrium.

\subsection{Existence of equilibrium}

We prove now that $(x^{\ast },p^{\ast })$ is a Walrasian equilibrium in the cases considered in the Theorem. In all cases we prove that, for each agent $i$, either $m^i(p^\ast)>0$ or else the null bundle 0 is the only affordable one. A standard argument follows converting such quasiequilibrium into an equilibrium. Suppose $m^i(p^\ast)>0$ and there exists $y^i\in C_i$ such that $u^i(y^i)>u^i(x^{i\ast})$ and $p^*\cdot y^i\leq m^i(p^\ast)$. Then, for $\lambda<1$ sufficiently close to 1, $\lambda y^i\in C_i$, $p^*\cdot \lambda y^i< m^i(p^\ast)$ and, by continuity of preferences, $u^i(\lambda y^i)>u^i(x^{i\ast})$, contradicting $x^{i\ast}\in \dminexp^i(p^*)$. For the remaining case $m^i(p^\ast)=0$, if 0 is the sole affordable bundle, then $0=x^{\ast i}$ trivially is $i$'s optimal choice subject to her budget constraint.

In all cases, we skip the possibility of $\sum_ie^i(v^i,p^*)<p^*\cdot Q$. By definition of $m^i$, all agents $i$ would have $m^i(p^*)>e^i(v^i,p^*)$ so they would certainly be satiated following the definition of $\dminexp^{i}$. Therefore we would trivially have an equilibrium. Note that, by skipping such a possibility, we have $p^{\ast }\cdot x^{i\ast }=m^i(p^\ast)$ for all $i$.

We first consider the case $\min_ic^i>\sum_lq_l$. We show that $p^{\ast }>>0$. Suppose, by way of contradiction, that $p^*_l=0$ for some good $l$. Since $p^*\cdot Q>0$ and $\sum_{i}x^{i\ast }-Q=0$, there must be an individual $j$ with $p^{\ast }\cdot x^{j\ast }=m^j(p^\ast)>0$. Take a vector $\delta $ containing zeros in all coordinates but $l$, where it contains $\epsilon>0$. Notice that $c^j>\sum_lq_l$ implies  $x^{j\ast }+\delta\in C^i$ for $\epsilon$ small enough. By monotonicity and continuity of preferences, and since $x^{j\ast }+\delta$ is also affordable, a standard argument shows that $x^{j\ast }$ is not a quasiequilibrium allocation for $j$ under prices $p^{\ast }$. We conclude that $p^{\ast }>>0$. Consequently, for each agent $i$ with $m^i(p^\ast)=0$, the 0 bundle is her only affordable bundle.

We now consider the case when both Inada condition $u^i(x^i)=u^i(0)$ unless $x^i>>0$ and $\tu^i>u^i(0)$ hold for all $i$. $p^\ast\in\Delta^L$ contains at least one strictly positive price, thus $m^i(p^\ast)\geq e^i(\tu^i,p^\ast)>0$ for all $i$.

Lastly, we consider the case that a common favorite object $l$ and a strictly positive IR allocation $\tilde{x}$ both exist. We argue that $p_{l}^{\ast }>0.$ Suppose that $%
p_{l}^{\ast }=0.$ Since $p^{\ast }\in \Delta^L ,$ there must be an object $k\neq l
$ with $p_{k}^{\ast }>0.$ For any agent $i$ who is consuming object $k$, substituting his consumption of object $k$ for an equal consumption of
object $l$ saves expenses and increases utility. Hence $x^{i\ast }\notin
\dminexp^i (p^{\ast })$. This contradiction shows that $p_{l}^{\ast}>0.$
Notice that, in consequence, $e^{i}(v^{i},p^{\ast })>0$ for all $i$. 

Our next step is to establish that $\ul m=\min\{ m^i(p^*) : 1\leq i \leq
I\}>0$. First, if $\ul m=\min\{ e^i(v^i,p^*) : 1\leq i \leq I\}$ then we
are done because $e^i(v^i,p^*)>0$ for all $i$. Ruling out this case,
there must exist $i$ with $m^i(p^*)<e^i(v^i,p^*)$, which implies that
$\sum_i m^i(p^*) = p^*\cdot Q=\sum_i p^*\cdot \tilde{x}^i$. Now, if
\begin{equation}
  \label{eq:ulm}
\ul m=\min\{ e^i(\tu^i,p^*): 1\leq i \leq I\},
\end{equation}
then there is $h$ with $\ul m =
m^h(p^*) \leq e^i(\tu^i,p^*)$ for all $i$; which implies by the
definition of the income functions that $m^i(p^*) =
e^i(\tu^i,p^*)$ for all $i$. 
But $e^i(\tu^i,p^*)\leq p^*\cdot \tilde{x}^i$ for all $i$ and 
\[
\sum_i e^i(\tu^i,p^*) = \sum_i m^i(p^*) =  p^*\cdot Q=\sum_i p^*\cdot\tilde{x}^i  \]
imply that
$e^i(\tu^i,p^*)= p^*\cdot \tilde{x}^i$  for all $i$. So $ m^i(p^*) = p^*\cdot \tilde{x}^i$ for all $ i $. Because $\tilde{x}^i>>0$, $m^i(p^*)>0$ for all $ i $. 

Finally, if Equation~\eqref{eq:ulm} does not hold, then
$0\leq \min\{ e^i(\tu^i,p^*): 1\leq i \leq I\}<\ul m$. So  $m^i(p^*)>0$ for all $ i $.

\subsection{Properties of a competitive equilibrium allocation $x^*$}
\subsubsection{Pareto optimality}
We disregard the case $\sum_i e^i(v^i, p^*) <
p^*\cdot Q$ in which clearly every agent $i$ is satiated since $m^i(p^*)>e^i(v^i, p^*)$. In the cases that follow below, any satiated agent $i$ must have $m^i(p^*)=e^i(v^i, p^*)$.
Suppose that $y^i\in C^i$ and that $u^i(y^i)\geq
u^i(x^{i *})$. Then we must have $p^*\cdot y^i\geq m^i(p^*)$ because
otherwise $p^*\cdot y^i< m^i(p^*)\leq e^{i}(v^{i},p^{\ast })$ and if $i$ is
satiated, then $p^*\cdot y^i< e^{i}(v^{i},p^{\ast })$ and 
$u^i(y^i)\geq v^i$ contradicts the definition of $e^i$; if $i$ is not satiated, there would exists $z^i$ with 
$p^*\cdot z^i< m^i(p^*)$ and $u^i(z^i)>u^i(y^i)\geq u^i(x^{i *})$.

Obviously if $y^i\in C^i$ and  $u^i(y^i)>
u^i(x^{i *})$ then $p^*\cdot y^i> m^i(p^*)$.
So if $y=(y^i)$ Pareto dominates $x^*$ then $\sum_i p^*\cdot y^i>\sum_im^i(p^*)$.
But by Equation~\eqref{eq:sumweqsummi} this is impossible if $y$ is an
allocation. 

\subsubsection{Individual rationality}

To show that $x^*$ is individually rational it suffices to notice that $m^{i}(p^{\ast})\geq e^{i}(\tu^{i},p^{\ast })$ for all $i$. 

\subsubsection{No justified envy}

Suppose that $i$ envies $j$ at $x^*$. This
implies that $i$ is not satiated, hence $m^{i}(p^{\ast })<
e^{i}(v^{i},p^{\ast })$. It also implies that $m^i(p^*)<m^j(p^*)$ as
$m^i(p^*)<p^*\cdot x^{j*} \leq m^j(p^*)$. By Lemma~\ref{lem:prepenvy}, then,
$m^j(p^*) = e^j(\tu^j,p^*)$.

We obtain that 
\[p^*\cdot x^{i*}=m^i(p^*)< m^j(p^*)=e^j(\tu^j,p^*),\] and hence
$u^j(x^i)<\tu^j$ by definition of expenditure function. So $i$'s
envy is not justified.

\section{Proof of Theorem~\ref{thm:existence}.3 and Corollary \ref{cor:lipschitz}}\label{sec:prooflipschitz}

\subsection{Proof of Corollary \ref{cor:lipschitz}}\label{sec:proofcorollary}
We denote a constant of Lipschitz continuity common to all utility functions by $\theta$. Let $y$ be an IR allocation, which exists by assumption.
Consider an additional object $e\notin O$, and an $\alpha$-extended economy, for any $\alpha\in (0,1)$, with the capacity vector $Q^\alpha$:
\begin{eqnarray*}
q^\alpha_l=(1-\alpha)q_l, \ \forall l \in O, \text{and } q^\alpha_e=\alpha\sum_ic_i.
\end{eqnarray*}
Preferences in this extended economy are defined to be:
\[U^i((x_l)_{l\in O},x_e)=u^i((x_l)_{l\in O})+\theta x_e. \]
Notice that under this construction, $e$ is a common favorite good in this extended economy.
By Lipschitz continuity, the allocation $y^\alpha$ with $y^{i\alpha}_l=(1-\alpha)y^i_l$ for $l\in O$ and $y^{i\alpha}_e=\alpha c_i$ meets $U^i(y^\alpha_i)>\tu^i$ for all $i\in I$. Therefore, by continuity of preferences, for $\beta>0$ low enough, the allocation $\beta Q^\alpha /N+(1-\beta)y^\alpha $ is a strictly positive IR allocation in the extended economy.

Therefore, by Theorem~\ref{thm:NJEeq}, each $\alpha$-extended economy contains a Pareto-optimal, IR and NJE allocation $x^\alpha$. We construct a sequence $(x^\alpha)_{\alpha}$ where $\alpha$ tends to zero. Wlog such sequence converges to some allocation $x^*$ (if not, a subsequence does.) Such limit is an allocation in the original economy.

$x^*$ is weak Pareto optimal. Suppose not, then there is an allocation $x'$ that strongly Pareto dominates $x^*$. Consider the allocation $x'^\alpha$ in the $\alpha$-extended economy where $x^{i\prime\alpha }_l=(1-\alpha)x^{i\prime}_l$ for $l\in O$ and $x^{ i\prime}_e=\alpha c_i$, for each $i\in I$. By continuity of preferences and for low enough $\alpha$, we have that $x'^\alpha$ strongly Pareto dominates $x^\alpha$. This contradicts that $x^\alpha$ is Pareto optimal.

$x^*$ is IR, since $U^i(x^\alpha_i)\geq \tu_i$ for all $i\in I$ and all $\alpha$.

$x^*$ has no strong justified envy. Suppose not. Then, some agent $i$ envies some other agent $j$ at $x^*$ and $u^j(x^{i*})>\tu^j$. For $\alpha$ low enough, and by continuity of preferences, $i$ envies $j$ at $x^\alpha$ and $U^j(x^{i\alpha})>\tu^j$. But this contradicts the fact that $x^\alpha$ satisfies NJE. 

\subsection{Proof of Theorem~\ref{thm:existence}.3}\label{sec:proofthirdstatement}
Following the above proof and noticing that linear utilities are Lipschitz continuous, we just need to show that, in this particular case, $x^*$ constructed above is Pareto optimal instead of weak Pareto optimal. 

Suppose that $x^{\ast }$ is not Pareto optimal. Let $x'$ be an allocation with $x'_e=0$ that Pareto dominates $x^*$. For any $\ep\in (0,1)$, consider
\[ x_{\ep}^{\prime \alpha} = x^\alpha + \ep (x'-x^*),
\]
and observe that $x_{\varepsilon}^{\prime \alpha}\rightarrow x^{\ast }+\varepsilon (x^{\prime }-x^{\ast })$ as $\alpha \rightarrow 0$.

By linearity of preferences, and the fact that $x'$ Pareto dominates $x^*$, 
we have that $U^{i}(x_{\varepsilon}^{i\prime  \alpha})\geq U^{i}(x^{i\alpha})$ for all $i$ with at least one strict inequality. We
have seen, by Theorem~\ref{thm:NJEeq}, that each $x^\alpha$ in the sequence is Pareto optimal in its corresponding $\alpha$-extended economy. Consequently, for any $\alpha$ and any $\ep\in (0,1)$, $x_{\varepsilon}^{\prime \alpha}$ 
\textit{cannot} be an allocation in its corresponding $\alpha$-extended economy.

Now, given that 
\begin{eqnarray*}
\sum_{i}x_{\varepsilon}^{i\prime  \alpha} &=&\sum_{i}x^{i \alpha}+\varepsilon
(\sum_{i}x^{i\prime }-\sum_{i}x^{i\ast }) \\
&=&Q^\alpha +\varepsilon (Q-Q)=Q^\alpha,
\end{eqnarray*}%
the market clearing aspect of being an allocation is met. So for $x_{\ep}^{\prime \alpha}$ not to be an allocation it must be the case that, for every $\varepsilon $ and $\alpha,$ there is at least one agent $i$ such that $x_{\varepsilon}^{i\prime  \alpha}\notin C^i.$

Observe that $x_{\varepsilon}^{i\prime  \alpha}\notin C^i$ means that we are in one of two cases  
\begin{enumerate}
\item \label{it:case1} $\sum_l x_{\varepsilon l}^{i\prime  \alpha}>c^i$, or 
\item \label{it:case2} $x_{\varepsilon l}^{i\prime  \alpha} < 0$ for some $ l $ (or both). 
\end{enumerate}

Moreover, note by definition of $x_{\ep}^{\prime \alpha}$ that if we are in case~\eqref{it:case1} then we are in \eqref{it:case1} for any $\ep'\geq \ep$, and  if we are in case~\eqref{it:case2} then we are in \eqref{it:case2} for any $\ep'\geq \ep$. 

Consider a sequence $\alpha_t\rightarrow 0$ and an associated sequence of corresponding equilibrium allocations $x^{\alpha_t}\rightarrow x^*$. Let $I^1_t$ be the set of agents $i$ for whom $x^{i \prime  \alpha_t}_{\ep}$ is in  case~\eqref{it:case1}  for all $\ep\in (0,1)$, and $I^2_t$ be the set of agents $i$ for whom $x^{i\prime  \alpha_t}_{\ep}$ is in  case~\eqref{it:case2}  for all $\ep\in (0,1)$.

Given that for any $\ep\in (0,1)$, no matter how small, there exists $i$ with
$x_{\ep}^{i\prime  \alpha_t}\notin C^i $ for all $\ep'\geq\ep$, and that the set of agents is finite, $I^1_t\cup I^2_t\neq \os$ for all $t$. 

Suppose first that $I^1_t\neq \os$ for infinitely many $t$.  Again, since the set of agents is finite, we can wlog assume that there exists a subsequence with the property that $I^1_t$ is invariant for all $t$ large enough. Let $I^*$ denote that invariant set. Select an agent $i^{\ast }\in I^{\ast }.$ Then 
\[c^{i^*}<\sum_l x^{ i^*\prime \alpha_t}_{\ep l } = \sum_l x^{i^* \alpha_t}_{l} + \ep (\sum_l x^{i^* \prime}_{l} - \sum_l x^{i^*\ast }_{l})
\] for all $\ep$ means that $\sum_l x^{i^* \alpha_t}_{l} =c^{i^*}$, as $x^{i^* \alpha_t}_t\in C^{i^*}$. Since this is true for all $t$ large enough and $x^{ \alpha_t}\rightarrow x^*$, $\sum_l x^{i^* \ast}_{l} =c^{i^*}$. Now, we must have $\sum_l x^{i^* \prime}_{l} - \sum_l x^{i^* \ast}_{l} > 0$, which implies that $\sum_l x^{i^* \prime}_{l} >c^{i^*}$, contradicting that $x'$ is an allocation. 

Suppose now that $I^1_t= \os$ for all but finitely many $t$. Then  $I^2_t\neq \os$ for infinitely many $t$.  Again, since the set of agents is finite, we can wlog assume that there exists a subsequence with the property that $I^2_t=I^*\neq \os$ for all $t$ large enough.  Select an agent $i^{\ast }\in I^{\ast }.$ Using the finiteness of the number of objects, there exists $l$ with the property that for all $t$ large enough, \[ \forall\ep\in (0,1),x^{i^*\alpha_t}_{l} + \ep (x^{i^* \prime}_{l} - x^{i^*\ast  }_{l} ) =  x^{\prime i^* \alpha_t}_{\ep l } < 0.\] Given that $x^{i^*}_t\in C_i$, this can only be true for all $\ep$ if $x^{i^* \alpha_t}_{l } = 0$, and $x^{i^* \prime}_l - x^{i^*\ast }_l< 0$. Now $x^{\alpha_t}\rightarrow x^*$ means that $x^{i^*\ast }_{l } = 0$, so we have $x^{i^* \prime}_l< 0$, which contradicts that $x'$ is an allocation. 

\section{Proof of Theorem~\ref{thm:cyclicenvy}}\label{sec:proofcyclicenvy}

\subsection{Theorem~\ref{thm:existence}}
To prove that the first two statements of Theorem~\ref{thm:existence} hold as before when NJE is extended, we omit the steps in common and highlight the differences from the previous proof. Let $\bla^i$ be the set of all $\la\in\Delta$ at which $i$ has no $\ep$-justified envy by exchange towards any agent at $\phi(\la)$. We prove that the collection of $\bla^i$ is still a KKM covering of $\Delta$.

\begin{lemma}\label{lem:Ciclosedcyclic}
	For every $ i\in I $, $\bla^{i}$ is closed.
\end{lemma}
\begin{proof}
	Let $\la_n$ be a sequence in $\bla^{i}$ such that $\la_n\rightarrow
	\la\in\Delta$. Let $x_n= \phi(\la_n)$. By continuity of $ \phi $,
	$x_n\rightarrow  x=\phi(\la)\in\A^*$. Now we prove that $\la\in 
	\bla^{i}$, that is, $ i $ does not have $\ep$-justified envy by exchange
	towards any other agent at $\phi(\la)$. Suppose that this is not the case. Then $i$ has $\ep$-justified envy by exchange towards some agent $j$, with the sequence $(i_k)_{k=1}^K$ being as in the definition of such envy.  By continuity of utility, and since the sequence $(i_k)_{k=1}^K$ is finite, for $n$ large enough we have $u^{i_k}(x^{i_{k+1}}_n) > u^{i_k}(x^{i_k}_n)$ for $ 1\le k\le K-1 $ while $u^{i_K}(x^{i_1}_n)>\tu^{i_K}-\ep$. So $i$ has $\ep$-justified envy by exchange towards $j$ at $x_n$, which is a contradiction. Therefore, $ \la\in \bla^{i}  $ and $\bla^i$ is closed.
\end{proof}

\begin{lemma}\label{lem:KKMconditionscyclic}
	For every $ \la\in \Delta $, $ \la \in \cup_{i\in \supp(\la)} \bla^i $.
\end{lemma}
\begin{proof}
	Suppose, towards a contradiction, that for some $ \la \in \Delta$,
	$\la \notin \cup_{i\in \supp(\la)} \bla^i  $. Let $x=\phi(\la)$. Then
	for every $i\in \supp (\la)$, there exists some $j$ such that $ i $ has $\ep$-justified envy by exchange towards $j$ at $x$. Suppose first that there exists such $j$, with corresponding sequence $(i_k)_{k=1}^K$, in which $\la^{i_K}=0$.  Let $y$ be the allocation obtained from $x$ by letting each $i_k$ get $x^{i_{k+1}}$ ($ 1\le k\le K-1 $) and $i_K$ get $x^i$. Clearly $y$ is $\ep$-IR, as $x$ was $\ep$-IR and
	$u^{i_K}(x^i) > \tu^{i_K}-\ep$. Note that $\la^{i_K}=0$ and $u^{i_k}(y^{i_k})>u^{i_k}(x^{i_k})$ for all $ 1\le k\le K-1 $ imply that $\sum_{h\in I} \la^h u^h(x^h) < \sum_{h\in I} \la^h
	u^h(y^h)$. We also have that  $\sum_{h\in I} \norm{x^h - \one} = \sum_{h\in I} \norm{y^h -\one}$, hence \[
	\sum_{h\in I} \la^h u^h(x^h) -\delta \sum_{h\in I} \norm{x^h - \one} < \sum_{h\in I} \la^h u^h(y^h) -\delta  \sum_{h\in I} \norm{y^h -\one},\] 
	 which contradicts the definition of $ x=\phi(\la)$.
	
	The above argument means that every $i\in \supp (\la)$ has $\ep$-justified 
	envy by exchange towards some agent $ j $, with corresponding sequence $(i_k)_{k=1}^K$ in which $\la^{i_K}>0$. Thus, $ i_K\in \supp (\la) $. But it means that $ i_K $ also has $\ep$-justified 
	envy by exchange towards some agent $ j' $, with corresponding sequence $(i'_k)_{k=1}^K$ in which $\la^{i'_K}>0$. Since the set of agents in $\supp(\la)$ is finite, there must exist a subset of agents $\{h_1,\ldots h_M\}\subseteq \supp(\la)$ such that $h_1$ has $\ep$-justified envy by exchange towards some agent with $ h_2 $ being the end of the corresponding sequence, $h_2$ has $\ep$-justified envy by exchange towards some agent with $h_3$ being the end of the corresponding sequence, and so on until $h_M$ has $ \ep $-justified envy by exchange towards some agent with $h_1$ being the end of the corresponding sequence. We write this situation as the following cycle
	\[
	h_1\rightarrow \cdots \rightarrow h_2 \rightarrow \cdots \rightarrow h_3 \rightarrow \cdots \rightarrow \cdots \rightarrow  h_M\rightarrow \cdots \rightarrow h_1,
	\]where $ a\rightarrow b $ means that $ a $ envies $ b $, and $ h_k\rightarrow \cdots \rightarrow h_{k+1} $ is the corresponding sequence of $ h_k $'s $\ep$-justified 
	envy by exchange towards some agent. Now note that if an agent $h$ appears more than once in the above cycle, we can shorten the cycle by skipping the agents between any two consecutive positions of $ h $ in the cycle. So we can, without loss of generality, focus on the cycle in which  each agent appears once. If we carry out the exchange in the cycle as in the proof of Lemma~\ref{lem:KKMconditions}, then we obtain an improvement on the objective that defines $\phi$. This is a contradiction.
\end{proof}

The remaining part of the proof is same as before.

\subsection{Theorem~\ref{thm:NJEeq}} 
Let $(x,p)$ be an equilibrium in Theorem~\ref{thm:NJEeq}. Suppose some agent $i$ has justified envy by exchange towards some agent $j$, with the sequence $(i_k)_{k=1}^K$ being as in the definition of such envy. By our construction of income functions, $m^{i}(p)<m^{j}(p)<m^{i_2}(p)<\cdots<m^{i_K}(p)$. So it must be that $m^{i_K}(p)=e^{i_K}(\tu^{i_K},p)$. It means that $x^i$ is not acceptable to $i_K$, which is a contradiction. So $x$ satisfies no justified envy by exchange. Then the third statement of Theorem~\ref{thm:existence} can be proved as before.

\bibliographystyle{econometrica}
\bibliography{envy}

\end{document}